\documentclass[english]{article}
\usepackage[T1]{fontenc}
\usepackage[latin9]{inputenc}
\usepackage{amsmath,amsthm} 
\usepackage{graphicx}
\usepackage{amssymb}
\usepackage{enumitem}
\usepackage{natbib} 
\usepackage{bussproofs}
\usepackage[letterpaper]{geometry} 
\usepackage{thmtools,thm-restate}
\usepackage{tikz}
\usepackage{pgf} 
\usetikzlibrary{patterns,automata,arrows,shapes,snakes,topaths,trees,backgrounds,positioning,through,calc}
\usepackage{setspace}
\usepackage{multicol}
\usepackage{graphicx}
\usepackage{stmaryrd}  
\usepackage{comment}
\usepackage{hyperref}
\usepackage{color}
\xdefinecolor{darkgreen}{rgb}{0,0.6,0}
\xdefinecolor{darkyellow}{rgb}{1,.6,.2}

\usepackage{natbib}

\theoremstyle{definition}

\newtheorem{proposition}{Proposition}
\newtheorem{corollary}{Corollary}
\newtheorem{example}{Example}
\newtheorem{definition}{Definition}
\newtheorem{lemma}{Lemma}

\usepackage[page,toc,titletoc,title]{appendix}


\onehalfspace
 
\begin{document} 

 \title{Stable Voting}
 
 \title{Stable Voting}
 \author{Wesley H. Holliday$^\dagger$ and Eric Pacuit$^\ddagger$ \\ \\
 $\dagger$ University of California, Berkeley {\normalsize (\href{mailto:wesholliday@berkeley.edu}{wesholliday@berkeley.edu})} \\
 $\ddagger$ University of Maryland {\normalsize (\href{mailto:epacuit@umd.edu}{epacuit@umd.edu})}}
 
 \date{February 11, 2023 \\
 {\small Extended version of article forthcoming in \textit{Constitutional Political Economy}.}}
 
\maketitle

 \begin{abstract} \noindent We propose a new single-winner voting system using ranked ballots: Stable Voting. The motivating principle of Stable Voting is that if a candidate $A$ would win without another candidate $B$ in the election, and $A$ beats $B$ in a head-to-head majority comparison, then $A$ should still win in the election with $B$ included (unless there is another candidate $A'$ who has the same kind of claim to winning, in which case a tiebreaker may choose between such candidates). We call this principle Stability for Winners (with Tiebreaking). Stable~Voting satisfies this principle while also having a remarkable ability to avoid tied outcomes in elections even with small~numbers~of~voters.  
\end{abstract}

\section{Introduction}\label{Introduction}
 
Voting reform efforts in the United States have achieved significant recent successes in replacing Plurality Voting with Instant Runoff Voting (IRV) for major political elections, including the 2018 San Francisco Mayoral Election and the 2021 New York City Mayoral Election. It is striking, by contrast, that \textit{Condorcet voting methods} are not currently used in any political elections.\footnote{The Condorcet voting method Nanson was used in Marquette, Michigan, in the 1920s (\citealt[p.~491]{Hoag1926}). We know of no cities using Condorcet methods since then, but see the Condorcet Canada Initiative at \url{https://condorcet.ca}.} Condorcet methods use the same ranked ballots as IRV but replace the counting of first-place votes with head-to-head comparisons of candidates: do more voters prefer candidate $A$ to candidate $B$ or prefer $B$ to $A$? If there is a candidate $A$ who beats every other candidate in such a head-to-head majority comparison, this so-called \textit{Condorcet winner} wins the election. 

As of July 2021, the Ranked Choice Voting Election Database maintained by FairVote\footnote{We are grateful to Deb Otis at FairVote for granting us access to the database.} contains 149 IRV elections from the United States for which the existence of a Condorcet winner can be determined from public data; and in every election, there was a Condorcet winner. Thus, in all of these elections, a Condorcet method would settle the election simply by the identification of the Condorcet winner,\footnote{Of course, had the official voting method been a Condorcet method, voters might have voted differently, perhaps leading to no Condorcet winner in one of these elections.} rather than the many rounds of iterative elimination of candidates and transferring of votes involved in some complicated IRV calculations. Moreover, in all but one of the elections, IRV chose the Condorcet winner anyway. Indeed, proponents of IRV claim that one of its advantages is that it almost always elects the Condorcet winner (\citealt{Dennis2008,Dennis2018}). The one case in the database in which IRV did not elect the Condorcet winner---the 2009 Mayoral Election in Burlington, Vermont, discussed below---was a source of controversy (see, e.g., \citealt{Gierzynski2009}, \citealt{Bouricius2009}). All of this raises the question: why not always simply elect the Condorcet~winner?

The problem is that a Condorcet winner is not guaranteed to exist (see Example \ref{NoCondorcet} below). Thus, some backup plan must be in place for how to elect a winner when there is no Condorcet winner. Proponents of IRV argue that ``The need for a backup plan has led to quite complicated Condorcet systems, perhaps the most popular being Schwartz Sequential Dropping,\footnote{Also known as the Schulze or Beat Path method (\citealt{Schulze2011}), discussed below. } which are ultimately more opaque and difficult to explain than IRV'' (\citealt{Dennis2008}) and that  ``If your voters are mathematically inclined, then consider Condorcet voting because it has some mathematical advantages but is more complicated to understand'' (\citealt{ONeill2021}). There are at least two responses to these claims about Condorcet methods being more complicated. First, there are Condorcet methods that are simple---or at least simpler than IRV. Take the Minimax (or Simpson-Kramer) method: the winner is the candidate whose worst head-to-head loss is smallest. Second, even if the backup plan is more complicated than that, we must consider which of the following approaches is preferable: (1) deciding almost every election in a simple way---just elect the Condorcet winner---and rarely applying a more complicated backup plan, perhaps more complicated than IRV, or (2) deciding many elections with a fairly complicated iterative elimination of candidates and transferring of votes, which may cause controversy when failing to elect a candidate who beats every other?

Moreover, there may be no objection to the use of sophisticated Condorcet systems in some contexts, including some elections in committees, clubs, etc.   For such contexts, we propose a new Condorcet voting method---hence a new backup plan when there is no Condorcet winner---that we call Stable Voting. Rather than aiming for the simplest possible backup plan, Stable Voting aims to generalize the following special property of Condorcet winners to all~winners:
\begin{itemize}
\item[$\bullet$] if a candidate $A$ would be the Condorcet winner without another candidate $B$ in the election, and $A$ beats $B$ in a head-to-head majority comparison, then $A$ is still the Condorcet winner in the election with $B$ included.
\end{itemize}
This is a kind of \textit{stability} property of Condorcet winners: you cannot dislodge a Condorcet winner $A$ by adding a new candidate $B$ to the election if $A$ beats $B$ in a head-to-head majority vote. For example, although the 2000 U.S. Presidential Election in Florida did not use ranked ballots, it is plausible (see \citealt{Magee2003}) that Al Gore ($A$) would have won without Ralph Nader ($B$) in the election, and Gore would have beaten Nader head-to-head. Thus, Gore should still have won with Nader included in the election. Indeed, Gore was plausibly the Condorcet winner with and without Nader in the election. Infamously, however, Plurality Voting chose George W. Bush when Nader was included.

The most obvious generalization of the special property of Condorcet winners above to all winners selected by some voting method is what we call \textit{Stability for Winners}\footnote{This principle, studied in \citealt{HP2020}, is almost the same as a principle mentioned in passing by Bordes \citeyearpar{Bordes1983}: ``you cannot turn $x$ into a loser by introducing new alternatives to which $x$ does not lose in duels'' (p.~125). Voting methods that satisfy Stability for Winners include Top Cycle (a.k.a.~GETCHA, as in \citealt{Schwartz1986}), Banks (\citealt{Banks1985}), the Uncovered Set (\citealt{Miller1980}), and Split Cycle (see \citealt{HP2020}).}:
\begin{itemize}
\item[$\bullet$] if a candidate $A$ would be a winner without another candidate $B$ in the election, and $A$ beats $B$ in a head-to-head majority comparison, then $A$ should still be a winner  in the election with $B$ included.
\end{itemize}
A difficulty that arises with Stability for Winners as a generalization of the special property of Condorcet winners above is that in some cases it is incompatible with tiebreaking to select a \textit{single} winner. The reason is that there may be more than one candidate with the same claim to winning as $A$ has in the statement of Stability for Winners, i.e., there may be another $A'$ who  also would have won without some candidate $B'$ in the election whom $A'$ beats head-to-head.\footnote{In fact, we have the following impossibility theorem (\citealt{HPZ2022}): there is no anonymous and neutral voting method that satisfies Stability for Winners and selects a unique winner in every election that is \textit{uniquely weighted} as defined in Section~\ref{BenefitsSection}.} But this does not mean we should give up on the idea of stability. The fact that a candidate $A$ would have won without another candidate $B$ in the election whom $A$ beats head-to-head is a prima facie reason why $A$ should win in the election with $B$. That reason can be undercut if---and only if, we argue---there is another candidate $A'$ who has the same kind of claim to winning; in this case, it is legitimate for a tiebreaker to choose between $A$ and $A'$ and any other candidates with the same kind of claim~to~winning. 

Thus, we propose the modified \textit{Stability for Winners with Tiebreaking}:
\begin{itemize}
\item[$\bullet$] if a candidate $A$ would be a winner  without another candidate $B$ in the election, and $A$ beats $B$ in a head-to-head majority comparison, then $A$ should still be a winner  in the election with $B$ included, unless there is another candidate $A'$ with the same kind of claim to winning, in which case a tiebreaker may choose between $A$ and $A'$ and any other candidates with the same kind of claim to winning.
\end{itemize}
In particular, in our view there is one situation in which there is a straightforward justification for why $A$ loses with $B$ in the election, given that $A$ wins without $B$ and that $A$ beats $B$ head-to-head: there are candidates $A'$ and $B'$ such that $A'$ wins without $B'$,  $A'$ beats $B'$ head-to-head, and \textit{the margin by which $A'$ beats $B'$ is larger than the margin by which $A$ beats $B$}. 

The voting method that we propose in this paper, Stable Voting, satisfies Stability for Winners with Tiebreaking by design. It also has a remarkable ability to avoid tied outcomes in elections even with small numbers of voters. We define Stable Voting in Section~\ref{StableVotingSection} and compare it to other voting methods in several example elections in Section \ref{ApplicationSection}. In Section \ref{BenefitsSection}, we discuss what we take to be some of the benefits of Stable Voting, and in Section \ref{CostsSection}, we discuss some of its costs.  We conclude in Section \ref{ConclusionSection}. For an implementation of Stable Voting, see \href{https://github.com/epacuit/stablevoting}{github.com/epacuit/stablevoting}; for formal verification of the proofs in this paper, see \href{https://github.com/asouther4/lean-social-choice}{github.com/asouther4/lean-social-choice}; and to run elections with Stable Voting, visit \href{https://stablevoting.org}{stablevoting.org}.

\section{Stable Voting defined}\label{StableVotingSection}

In an election with ranked ballots, given distinct candidates $A$ and $B$, the \textit{margin of} $A$ vs.~$B$ is the number of voters who rank $A$ above $B$ minus the number who rank $B$ above $A$. For example, if 7 voters rank $A$ above $B$ and 3 rank $B$ above $A$, then the margin of $A$ vs.~$B$ is 4, and the margin of $B$~vs.~$A$~is~$-4$.\footnote{In our view, it is important to measure the sizes of wins and losses in terms of \textit{margin}. Another approach agrees that  $A$ beats $B$ head-to-head if and only if the margin of $A$ over $B$ is positive, but it says that the size of $A$'s win over $B$ is measured solely in terms of the number of voters who rank $A$ above $B$, ignoring the number of voters who rank $B$ above $A$. E.g., if 51 voters rank $A$ above $B$, 49  rank $B$ above $A$, 50  rank $X$ above $Y$, and 50  are indifferent between $X$ and $Y$, the alternative approach says that $A$'s win over $B$ is a bigger win than $X$'s (Pareto-dominating) win over $Y$. We reject this alternative approach.} 

If there is a Condorcet winner in an election, Stable Voting selects that Condorcet winner, so no further calculations are necessary. But Stable Voting also works when there is no Condorcet winner. First, we will define a simplified version of Stable Voting, which we call Simple Stable Voting.

For a set of ranked ballots, Simple Stable Voting selects a winner as follows:
\begin{itemize}
\item[$\bullet$] If only one candidate $A$ appears on all ballots, then $A$ wins.
\item[$\bullet$] Otherwise list all head-to-head matches of the form $A$ vs.~$B$ in order from the largest to smallest margin of $A$~vs.~$B$. Find the first match $A$ vs.~$B$ in the list such that $A$ wins according to Simple Stable Voting \textit{after $B$ is removed from all ballots}; this $A$ is the winner for the original set of~ballots.\footnote{\label{TiesNote}If there are multiple matches with the same margin, then there may be more than one earliest match $A$~vs.~$B$ in the list  such that the first candidate wins after the second is removed from the ballots. In this case, all such $A$'s are tied Simple Stable Voting winners for the original set of ballots. See Section \ref{BenefitsSection} for discussion of the rarity of such ties.}
\end{itemize}
Such an $A$ is guaranteed to exist (see Proposition \ref{QRLem}).  See Examples \ref{NoCondorcet}, \ref{4cands3cycles}, and \ref{EaddedEx} for concrete examples of determining the winner.

Note that our list of head-to-head matches includes those where $A$ has a \textit{negative} margin vs.~$B$. Indeed, we cannot ignore such matches; for example, in a three-candidate election where $X$ beats $Y$ head-to-head, $Y$ beats $Z$ head-to-head, and $Z$ beats $X$ head-to-head, there is no match $A$ vs.~$B$ with a positive margin such that $A$ wins after removing $B$.\footnote{Assuming that Majority Voting governs two-candidate elections.} However, the fact that $A$ would win without $B$ in the election, and $A$ \textit{loses} head-to-head to $B$, does not seem to provide even a prima facie reason why $A$ should win in the election with $B$ included. For this reason, Stable Voting departs from Simple Stable Voting by not including every match $A$ vs.~$B$ where the margin of $A$ vs.~$B$ is negative. But Stable Voting does include some matches with negative margins.

The key distinction is that when there is a \textit{majority cycle}---a list of candidates each of whom beats the next head-to-head and the last of which beats the first, as with $X,Y,Z$ above---not every loss can be considered a \textit{defeat} that disqualifies a candidate from winning. While $A$ may lose head-to-head to $B$, $A$ may beat a candidate $C$ who in turns beats $B$; in fact, the margins by which $A$ beats $C$ and $C$ beats $B$ may be at least as large as the margin by which $B$ beats $A$. In this case, we say that $B$ does not \textit{defeat} $A$. In general, we say that $B$ does not defeat $A$ if we can make a list of candidates starting with $A$ and ending with $B$ such that the margin of each candidate vs.~the next  in the list  is at least the margin of $B$ vs.~$A$; if there is no such list, then $B$ defeats $A$ (this notion of defeat is defended at length in \citealt{HP2020b}). Now, the fact that $A$ would win without $B$ in the election, and $B$ \textit{does not defeat $A$}, \textit{does} provide a prima facie reason why $A$ should win in the election with $B$ included---only to be undercut by the tiebreaking considerations in Section \ref{Introduction}.

Thus, Stable Voting may be defined analogously to Simple Stable Voting except with the following proviso:
\begin{itemize}
\item[$\bullet$] When we reach a match $A$ vs.~$B$ with a negative margin, we consider such a match only if $B$ does not defeat $A$.
\end{itemize}
In fact, the proviso applies uniformly to matches regardless of margin, but in the case of non-negative margins there is no need to check it, since if $A$ has a non-negative margin vs.~$B$, then $B$ does not defeat $A$. 

It can be proved that the Stable Voting winner is always \textit{undefeated}, i.e., not defeated by anyone.\footnote{The proof is by induction on the number of candidates. If $A$ wins, then there is some $B$ who does not defeat $A$ such that $A$ wins without $B$ in the election. Hence by the inductive hypothesis, $A$ is undefeated in the election without $B$. Then since $B$ does not defeat $A$ in the full election, it is easy to see from the definition of defeat that $A$ is still undefeated.} As a result, Stable Voting may also be defined with a modified first clause and an even shorter list of matches:
\begin{itemize}
\item[$\bullet$] If there is a single undefeated candidate $A$, then $A$ wins.
\item[$\bullet$] Otherwise list all head-to-head matches of the form $A$ vs.~$B$, where $A$ is undefeated, in order from the largest to the smallest margin of $A$ vs.~$B$. Find the first such that $A$ wins according to  Stable Voting \textit{after $B$ is removed from all ballots}; this $A$ is the winner for the original set of~ballots.
\end{itemize}
Thus, Stable Voting may be seen as a way of breaking ties between the undefeated candidates (and hence as a refinement of the Split Cycle rule in \citealt{HP2020} that finds all undefeated candidates). 

In addition to the important theoretical distinction between Simple Stable Voting and Stable Voting, there is a practical one: our current implementation of Stable Voting is considerably faster than that of Simple Stable Voting for elections with many candidates. On the other hand, interestingly, in simulations these methods select the same winners nearly 100\% of the time.\footnote{In fact, it is open whether Stable Voting and Simple Stable Voting ever select different winners for elections that are \textit{uniquely weighted} as defined in Section \ref{BenefitsSection}. Moreover, if we handle a tie between margins (recall Footnote \ref{TiesNote}) by declaring $A$ a winner if there is some way of breaking the tie between margins so that $A$ wins (so-called parallel universe tiebreaking), then it is open whether the methods ever select different winners.}Since we wish to illustrate both methods, we will not use the modified first clause or pare down the list of matches as above in the examples in Section~\ref{ApplicationSection}.

\section{Stable Voting applied}\label{ApplicationSection}

With only two candidates, Stable Voting is the same as Majority Voting. For if $A$ beats $B$ head-to-head, then $A$ vs.~$B$ (with positive margin) appears on the list of matches by descending margin before $B$ vs.~$A$ (with negative margin), and $A$ wins after $B$ is removed from all ballots, as $A$ is the only candidate left. 

With more than two candidates, Stable Voting becomes more interesting.\footnote{Although it will not affect the following examples, if so desired, one can modify Stable Voting with the rule that if there is a unique candidate with no losses, called a \textit{weak Condorcet winner}, then we elect that candidate; and if there are several, then we apply the Stable Voting procedure but only to matches $A$ vs.~$B$ where $A$ is a weak Condorcet winner, thereby guaranteeing that the winner will be among the weak Condorcet winners.}

\begin{example}\label{VermontExample} In the 2009 Mayoral Election in Burlington, Vermont, the progressive candidate Bob Kiss was elected using IRV. However, checking the head-to-head matches of the candidates revealed that the Democrat Andy Montroll beat Bob Kiss and every other candidate head-to-head, as shown on the left of Figure \ref{VermontGraph}.\footnote{\label{Unranked}For the purposes of this Condorcet analysis, when a voter submits a  linear order of a proper subset of the candidates, we regard this ballot as a strict weak order of the entire set of candidates by placing all unranked candidates in a tie below all ranked candidates. This is also the methodology adopted by FairVote for their Condorcet analyses.} Thus, Montroll was the Condorcet winner and hence the Stable Voting winner. Had the Republican, Kurt Wright, who lost the IRV election not been included in the election, then (other things equal) the IRV winner would have been Montroll. Thus, Wright's inclusion in the IRV election spoiled the election for Montroll. For further discussion of this spoiler effect for IRV, see \url{https://www.electionscience.org/library/the-spoiler-effect} (accessed 7/22/2021).

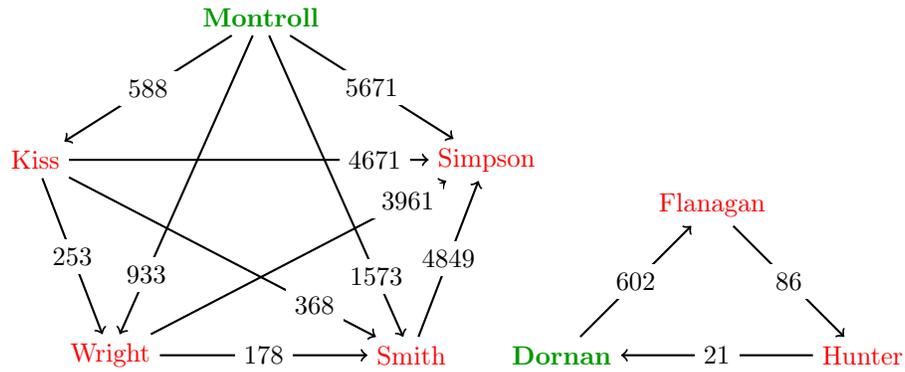
\begin{figure}[h]

\begin{center}

\begin{tikzpicture}

\node at (-1,1.1) (c) {\textcolor{red}{Kiss}}; 
\node at (5,1.1) (a) {\textcolor{red}{Simpson}}; 
\node  at (0,-1.5) (d) {\textcolor{red}{Wright}}; 
\node at (2,3) (b) {\textcolor{darkgreen}{\textbf{Montroll}}}; 
\node at (4,-1.5) (e) {\textcolor{red}{Smith}}; 

\path[->,draw,thick] (b) to node[fill=white] {$5671$} (a);
\path[->,draw,thick] (b) to node[fill=white] {$588$} (c); 
\path[->,draw,thick] (c) to node[fill=white,pos=.85] {$4671$} (a);
\path[->,draw,thick] (c) to node[fill=white] {$253$} (d); 
\path[->,draw,thick] (b) to node[fill=white,pos=.8] {$933$} (d);
\path[->,draw,thick] (d) to node[fill=white] {$178$} (e); 
\path[->,draw,thick] (c) to node[fill=white,pos=.8] {$368$} (e);
\path[->,draw,thick] (e) to node[fill=white] {$4849$} (a); 
\path[->,draw,thick] (d) to node[fill=white,pos=.87] {$3961$} (a);
\path[->,draw,thick] (b) to node[fill=white,pos=.8] {$1573$} (e);

\node at (6,-1.5) (D) {\textcolor{darkgreen}{\textbf{Dornan}}}; 
\node at (8,.5) (F) {\textcolor{red}{Flanagan}}; 
\node at (10,-1.5) (H) {\textcolor{red}{Hunter}}; 

\path[->,draw,thick] (D) to node[fill=white] {\textcolor{black}{$602$}} (F);
\path[->,draw,thick] (F) to node[fill=white] {\textcolor{black}{$86$}} (H);
\path[->,draw,thick] (H) to node[fill=white] {$21$} (D);

\end{tikzpicture}

  \end{center}
  \caption{Left: head-to-head results for the 2009 Burlington Mayoral Election. An arrow from candidate $A$ to candidate $B$ indicates that $A$ beat $B$ head-to-head. The arrow is labeled with $A$'s margin of victory vs.~$B$. Right: head-to-head results for candidates in the Smith set of the 2007 Glasgow City Council election for Ward 5 (Govan).}\label{VermontGraph}
  \end{figure}
\end{example}

\begin{example}\label{NoCondorcet} For an example with no Condorcet winner, we consider the 2007 Glasgow City Council election for Ward 5 (Govan), available in the Preflib database (\citealt{Mattei2013}). Each voter submitted a linear order of some subset of the eleven candidates, which we convert to a strict weak order with all unranked candidates tied at the bottom (as in Footnote \ref{Unranked} and Preflib's file 00008-00000009.toc). The election was run using Single-Transferable Vote to elect four candidates.\footnote{See \url{https://en.wikipedia.org/wiki/2007_Glasgow_City_Council_election}.} However, we can also run single-winner voting methods on these ballots. Plurality and IRV select Allison Hunter, while Plurality with Runoff selects John Flanagan.\footnote{For Plurality with Runoff, Hunter and Flanagan advance to the runoff, and then Flanagan beats Hunter by 86 votes.} In this election, there is no Condorcet winner. The \textit{Smith set}---the smallest set of candidates each of whom beats each candidate outside the set head-to-head---consists of three candidates, Stephen Dornan, Flanagan, and Hunter, as shown on the right of Figure~\ref{VermontGraph}, in a \textit{majority cycle}: Dornan beats Flanagan by 602, Flanagan beats Hunter by 86, and Hunter beats Dornan by~21. 

Stable Voting determines the winner by going down the list of matches:
\begin{enumerate}
\item Dornan vs. Flanagan: margin of $602$. \\
Dornan loses (to Hunter) after removing Flanagan. Continue to next match:
\item Flanagan vs. Hunter: margin of $86$.  \\
Flanagan loses (to Dornan) after removing Hunter. Continue to next match:
\item Hunter vs. Dornan: margin of $21$.  
\\ Hunter loses (to Flanagan) after removing Dornan. Continue to next match:
\item Dornan vs. Hunter: margin of $-21$.  \\ 
(Dornan is not \textit{defeated} by Hunter: Dornan beats Flanagan by 602, who beats Hunter by~86, and both of these margins are greater than 21.)\\
Dornan wins (against Flanagan) after removing Hunter. \textbf{Dornan is elected}.
\end{enumerate}
\end{example}
That Stable Voting elects Dornan in Example \ref{NoCondorcet} is a consequence of the following characterization of when a candidate is a Stable Voting winner in an election with three candidates.

\begin{proposition}\label{3cands} In any election with exactly three candidates, a candidate $A$ is a Stable Voting winner if and only if one of the following holds: 
\begin{enumerate}
\item $A$ has no head-to-head losses, and $A$'s maximal margin against another candidate is maximal among all margins between a candidate with no losses and another candidate;
\item  each candidate has a head-to-head loss, and the margin by which another candidate beats $A$ head-to-head is minimal among all positive margins.
\end{enumerate}
\end{proposition}
\begin{proof} First suppose that each candidate has a head-to-head loss. Since there are exactly three candidates in the election, it follows that each candidate also has a head-to-head win. Since each candidate has a head-to-head loss, if $A$ has a positive margin over $B$, then after removing $B$ from the ballots, $A$ loses to the third candidate, say $C$. Since each candidate also has a head-to-head win, if $A$ has a negative margin over $B$, then after removing $B$ from the ballots, $A$ wins against the third candidate. Hence the first matches  $A$~vs.~$B$ in the list of matches such that $A$ wins after removing $B$ from the ballots are the matches with the smallest negative margins; thus, the margin by which another candidate beats $A$ head-to-head is minimal among all positive margins.

Now suppose there is a candidate with no head-to-head-losses. If $A$ has no losses, then after removing any other candidate $B$ from the ballots, $A$ is among the winners; by contrast, if $A$ has a loss, then for any candidate $B$ such that $A$ has a non-negative margin over $B$, after removing $B$ from the ballots, $A$ loses to the third candidate. It follows that the first matches $A$ vs.~$B$ in the list of matches such that $A$ wins after removing $B$ are such that $A$ has no losses and the margin of $A$ over $B$ is maximal among all margins between a candidate with no losses and another candidate.\end{proof}

\begin{example}\label{4cands3cycles} Next we consider a hypothetical election with four candidates and no Condorcet winner. The head-to-head margins are shown on the left in Figure \ref{ExFigA}. The arrow from $A$ to $B$ labeled by $6$ indicates that the margin of $A$ vs.~$B$ is 6 and hence the margin of $B$ vs.~$A$ is $-6$. The arrow from $B$ to $C$ labeled by 4 indicates that the margin of $B$ vs.~$C$ is 4, and so on.

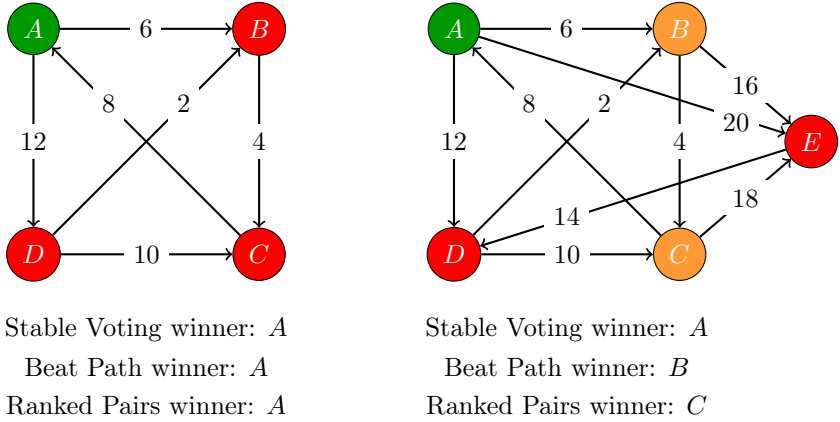
\begin{figure}[h]
\begin{center}
\begin{minipage}{2in}

\begin{center}

\begin{tikzpicture}

\node[fill=darkgreen,circle,draw, minimum width=0.25in] at (0,3) (a) {\textcolor{white}{$A$}}; 
\node[fill=red,circle,draw,minimum width=0.25in] at (3,3) (b) {$\textcolor{white}{B}$}; 
\node[fill=red,circle,draw,minimum width=0.25in] at (3,0) (c) {$\textcolor{white}{C}$}; 
\node[fill=red,circle,draw,minimum width=0.25in] at (0,0) (d) {$\textcolor{white}{D}$}; 

\path[->,draw,thick] (a) to node[fill=white] {\textcolor{black}{$6$}} (b);
\path[->,draw,thick] (b) to node[fill=white] {$4$}  (c);
\path[->,draw,thick] (c) to  (a);
\path[->,draw,thick] (a) to node[fill=white] {\textcolor{black}{$12$}}  (d);
\path[->,draw,thick] (d) to   (b);
\path[->,draw,thick] (d) to node[fill=white] {\textcolor{black}{$10$}}  (c);

\node[fill=white] at (2,2)  {$2$}; 
\node[fill=white] at (1,2)  {$8$}; 

\node at (1.5,-1) {Stable Voting winner: $A$};
\node at (1.5,-1.5) {Beat Path winner: $A$};
\node at (1.5,-2) {Ranked Pairs winner: $A$};

  \end{tikzpicture}
  \end{center}
  \end{minipage}\qquad\quad
  \begin{minipage}{2in}\begin{tikzpicture}

\node[fill=darkgreen,circle,draw, minimum width=0.25in] at (0,3) (a) {\textcolor{white}{$A$}}; 
\node[fill=darkyellow,circle,draw,minimum width=0.25in] at (3,3) (b) {\textcolor{white}{$B$}}; 
\node[fill=darkyellow,circle,draw,minimum width=0.25in] at (3,0) (c) {\textcolor{white}{$C$}}; 
\node[fill=red,circle,draw,minimum width=0.25in] at (0,0) (d) {\textcolor{white}{$D$}}; 
\node[fill=red,circle,draw,minimum width=0.25in] at (4.75,1.5) (e) {\textcolor{white}{$E$}}; 

\path[->,draw,thick] (a) to node[fill=white] {$6$} (b);
\path[->,draw,thick] (b) to node[fill=white] {$4$}  (c);
\path[->,draw,thick] (c) to  (a);
\path[->,draw,thick] (a) to node[fill=white] {\textcolor{black}{$12$}}  (d);
\path[->,draw,thick] (d) to   (b);
\path[->,draw,thick] (d) to node[fill=white] {$10$}  (c);

\path[->,draw,thick] (a) to  (e);
\path[->,draw,thick] (b) to node[fill=white] {\textcolor{black}{$16$}}  (e);
\path[->,draw,thick] (c) to node[fill=white] {\textcolor{black}{$18$}}  (e);
\path[->,draw,thick] (e) to   (d);

\node[fill=white] at (2,2)  {$2$}; 
\node[fill=white] at (1,2)  {$8$}; 
\node[fill=white] at (1.5,.5)  {$14$}; 
\node[fill=white] at (3.75,1.75)  {$20$}; 

\node at (1.5,-1) {Stable Voting winner: $A$};
\node at (1.5,-1.5) {Beat Path winner: $B$};
\node at (1.5,-2) {Ranked Pairs winner: $C$};

  \end{tikzpicture}

\end{minipage}
  \end{center}
  \caption{head-to-head results for elections in Examples \ref{4cands3cycles} and \ref{EaddedEx}}\label{ExFigA}
  \end{figure}
  
 Stable Voting determines the winner of the election by going down the list of matches (recall from Proposition \ref{3cands} that in a three-candidate election in which each candidate has one loss, the candidate with the smallest loss wins):
  \begin{itemize}
\item[$\bullet$] $A$ vs.~$D$: margin of 12.\\
$A$ loses (as $C$ wins) after removing $D$. Continue to next match:
\item[$\bullet$] $D$ vs.~$C$: margin of 10.\\
$D$ loses (as $A$ wins) after removing $C$. Continue to next match:
\item[$\bullet$] $C$ vs. $A$: margin of 8.\\
$C$ loses (as $D$ wins) after removing $A$. Continue to next match:
\item[$\bullet$] $A$ vs. $B$: margin of 6.\\
$A$ wins after removing $B$. \textbf{$A$ is elected.}  
\end{itemize}

\end{example}

\noindent For those familiar with Beat Path (\citealt{Schulze2011}) and Ranked Pairs (\citealt{Tideman1987}), we note that they also elect $A$ for the election on the left in Figure \ref{ExFigA}.

\begin{example}\label{EaddedEx} We now modify the election from Example \ref{4cands3cycles} by adding another candidate $E$, as shown on the right of Figure \ref{ExFigA}. Stable Voting determines the winner of the election by going down the list of matches:
\begin{itemize}
\item[$\bullet$] $A$ vs.~$E$: margin of 20.\\
$A$ wins after removing $E$ (see Example \ref{4cands3cycles}). \textbf{$A$ is elected}. 
\end{itemize}

\noindent By contrast, when the loser $E$ whom $A$ beats head-to-head joins the election, Beat Path and Ranked Pairs change their choices to $B$ and $C$, respectively, thereby violating Stability for Winners with Tiebreaking (see Definition \ref{StableDef}).\end{example}

\section{Benefits of Stable Voting}\label{BenefitsSection}

In this section, we discuss some of the beneficial properties of Stable Voting---as well as Simple Stable Voting, for which the properties also hold. A \textit{profile} $\mathbf{P}$ is an assignment to each voter in a finite electorate of a ranking (possibly incomplete) of the finitely many candidates in the election. A profile is \textit{uniquely weighted} if there are no distinct matches $A$ vs.~$B$ and $A'$~vs.~$B'$ with the same margin.  Given a candidate $B$ in $\mathbf{P}$, let $\mathbf{P}_{-B}$ be the profile obtained from $\mathbf{P}$ by removing $B$ from all rankings. We use `SV' to abbreviate `Stable Voting'.

\begin{proposition}\label{QRLem} Any profile $\mathbf{P}$ has at least one SV winner. Moreover, if $\mathbf{P}$ is uniquely weighted,  then there is only one SV winner in $\mathbf{P}$ (cf.~Footnote \ref{TiesNote}).
\end{proposition}
\begin{proof} By induction on the number of candidates. We claim there is some match $A$ vs.~$B$ in $\mathbf{P}$ such that  $A$ is an SV winner in $\mathbf{P}_{-B}$, and $B$ does not defeat $A$ in the sense of Section \ref{StableVotingSection}. For contradiction, suppose not. Consider any candidate $B_1$ in $\mathbf{P}$. In  $\mathbf{P}_{-B_1}$, there is an SV winner $B_2$ by the inductive hypothesis. By our supposition, $B_1$  defeats $B_2$. In $\mathbf{P}_{-B_2}$, there is an SV winner $B_3$ by the inductive hypothesis. By our supposition, $B_2$ defeats $B_3$. Continuing in this way, since there are only finitely many candidates, we obtain a list $B_k,\dots,B_n$ of candidates each of whom defeats the next such that $B_k=B_n$; but it is easy to see that such a list contradicts the definition of defeat. Now since there is some head-to-head match such that the first candidate wins after removing the second, and the second does not defeat the first, there is an earliest (i.e., first or tied for first) such match in the list of matches by descending margin. Hence there is at least one SV winner in $\mathbf{P}$. If $\mathbf{P}$ is uniquely weighted, there is a \textit{unique} first such match in the list, so there is only one SV winner.\end{proof}

The second part of Proposition \ref{QRLem} significantly undersells the power of SV to pick a unique winner. Figure \ref{ResolutenessFig} shows computer simulation results estimating the percentage of all linear profiles\footnote{A linear profile is a profile in which each voter submits a linear order of the candidates, disallowing ties. Simulations for non-linear profiles are of course also possible, as are simulations using non-uniform probability distributions on the set of profiles for a given number of candidates and voters. We will report on such simulation results in future work.} for a fixed number of candidates and voters in which a given voting method declares a tie between multiple winners, before any randomized tiebreaking. SV does much better than the other methods\footnote{For the version of IRV used in our simulations, see Footnote \ref{PUT}. We also tried a less standard version of IRV where if there are multiple candidates tied for the fewest first-place votes, \textit{all} are eliminated (unless this would eliminate all candidates). This leads to fewer ties than the more standard version of IRV, but SV still outperforms this version of~IRV.} at selecting a unique winner for profiles with around 5,000 voters or fewer. Moreover, for many other voting methods, for a fixed number of voters, the percentage of profiles producing a tie increases as we increase the number of candidates. Remarkably, for SV, the opposite happens: the percentage of profiles with a tie \textit{decreases} as we increase the number of candidates.

\begin{figure}[h]
\begin{center}
\includegraphics[scale=0.505]{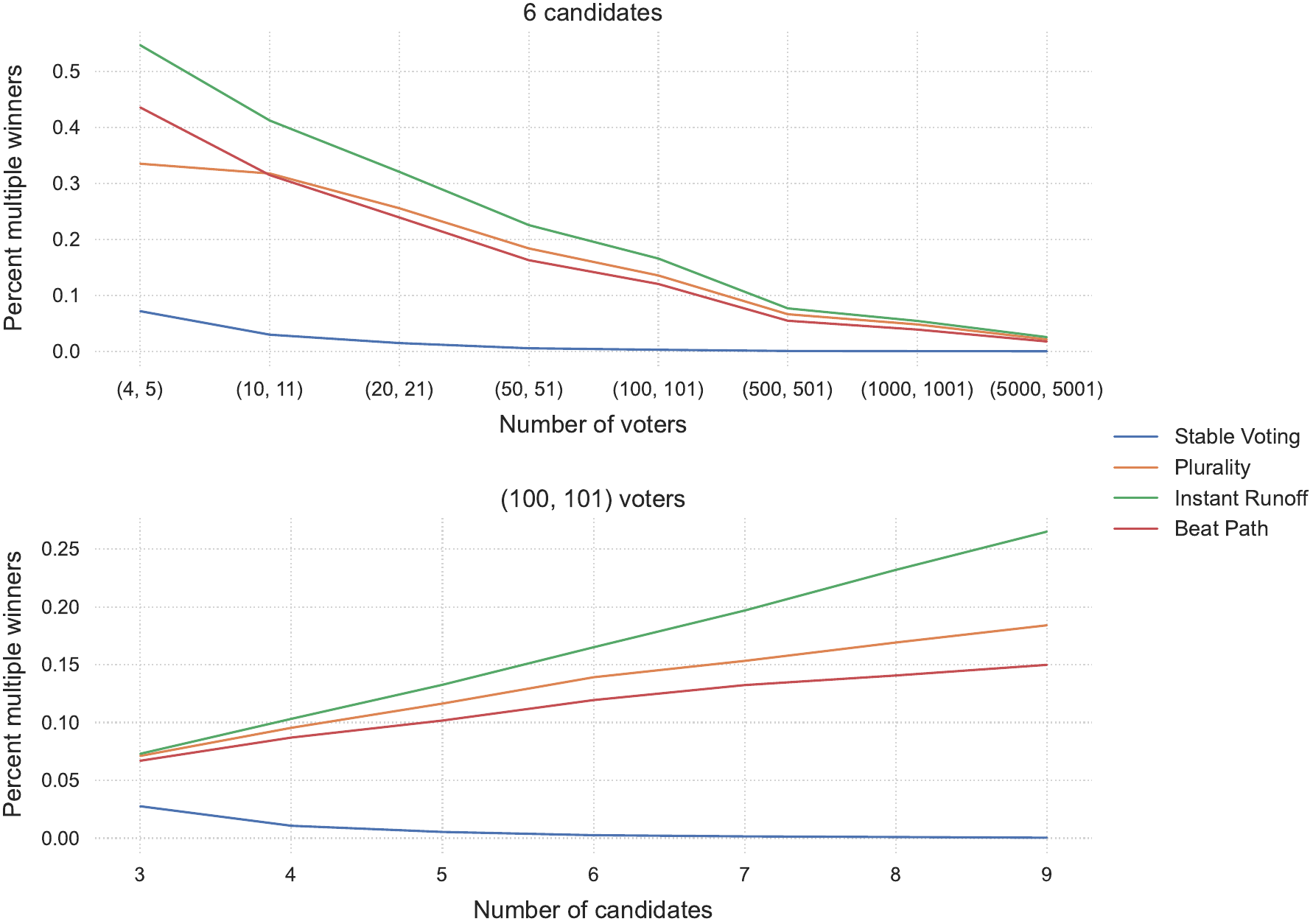}
\caption{estimates of the percentage of linear profiles for a given number of candidates and voters in which a given voting method selects multiple tied winners. For each data point, we randomly sampled 50,000 profiles for the even number of voters and 50,000 profiles for the odd number. We also obtained data for Minimax, but its graph coincides with that of Beat Path. We did not include Ranked Pairs due to the difficulty of computing (the anonymous, neutral, and deterministic version of) Ranked Pairs in non-uniquely-weighted profiles (\citealt{BrillFischer2012}).}\label{ResolutenessFig}
\end{center}
\end{figure}

Next we prove that Stable Voting satisfies the motivating principle discussed in Section \ref{Introduction}: Stability for Winners with Tiebreaking.

\begin{definition}\label{StableDef} Given a voting method $\mathsf{M}$ and profile $\mathbf{P}$, a candidate $A$ is \textit{stable for $\mathsf{M}$ in $\mathbf{P}$} if there is some $B$ in $\mathbf{P}$ whom $A$ beats head-to-head such that $A$ wins according to $\mathsf{M}$ in the profile $\mathbf{P}_{-B}$. We say that $\mathsf{M}$ satisfies \textit{stability for winners with tiebreaking} if for every profile $\mathbf{P}$, if there is a candidate who is stable for $\mathsf{M}$ in $\mathbf{P}$, then whoever wins in $\mathbf{P}$ according to $\mathsf{M}$ is stable for $\mathsf{M}$ in~$\mathbf{P}$.\end{definition}

\begin{proposition} SV satisfies stability for winners with tiebreaking.
\end{proposition}
\begin{proof} If some candidate $A$ is stable for SV in $\mathbf{P}$, this means there is a match $A$ vs.~$B$ with a \textit{positive} margin such that $A$ is an SV winner in $\mathbf{P}_{-B}$. It follows that if $X$ is any SV winner in $\mathbf{P}$, then the match $X$ vs.~$Y$ that witnesses this fact also has a positive margin. This implies that $X$ is stable for SV in $\mathbf{P}$.\end{proof}

 \noindent By contrast, Plurality, IRV, Minimax, Beat Path, and Ranked Pairs all violate the principle. That Beat Path and Ranked Pairs violate it can be seen by their choices of winners in Figure~\ref{ExFigA}.

Next we turn to the Condorcet criterion and a strengthening of it.

\begin{proposition} If a voting method $\mathsf{M}$ satisfies stability for winners with tiebreaking, then it satisfies the \emph{Condorcet criterion}:  if $A$ is a Condorcet winner in a profile $\mathbf{P}$, then $A$ is the unique winner according to $\mathsf{M}$ in $\mathbf{P}$.
\end{proposition}
\begin{proof} By induction on the number of candidates. If $A$ is a Condorcet winner in $\mathbf{P}$, then $A$ is the Condorcet winner in $\mathbf{P}_{-B}$ for every candidate $B\neq A$. Hence by the inductive hypothesis, $A$ is the unique winner in $\mathbf{P}_{-B}$ for every $B\neq A$. It follows that $A$ is the only candidate stable for $\mathsf{M}$ in $\mathbf{P}$, so by stability for winners with tiebreaking, $A$ is the unique winner according to $\mathsf{M}$.\end{proof}

\begin{corollary}\label{CondorcetLem} SV satisfies the Condorcet criterion.
\end{corollary}

Recall that the Smith set $Smith(\mathbf{P})$ of a profile $\mathbf{P}$ is the smallest set of candidates such that each candidate in the set beats each candidate outside the set head-to-head. The following proposition strengthens Corollary \ref{CondorcetLem}.

\begin{proposition}\label{SmithConsistency} SV satisfies the \textit{Smith criterion}: if $A$ is an SV winner in $\mathbf{P}$, then $A$ belongs to the Smith set of $\mathbf{P}$.
\end{proposition}
 
\begin{proof} We use a preliminary lemma about the Smith set, which is easy to check: if $A\in Smith(\mathbf{P}_{-B})$, and $B$ is not a Condorcet winner in $\mathbf{P}$, then $A\in Smith(\mathbf{P})$. Now we prove the proposition by induction on the number of candidates. Suppose $A$ is an SV winner in $\mathbf{P}$ as witnessed by $A$ vs.~$B$, so $A$ is an SV winner in $\mathbf{P}_{-B}$. Since $A$ is an SV winner in $\mathbf{P}$, it follows by Corollary~\ref{CondorcetLem} that $B$ is not a Condorcet winner. Since $A$ is an SV winner in $\mathbf{P}_{-B}$, by the inductive hypothesis, $A\in Smith(\mathbf{P}_{-B})$, so by the lemma, $A\in Smith(\mathbf{P})$.
\end{proof}

\begin{corollary} SV satisfies the \textit{Condorcet loser criterion}: if $A$ loses head-to-head to every other candidate in $\mathbf{P}$ (and there is more than one candidate in $\mathbf{P}$), then $A$ is not an SV winner in $\mathbf{P}$.
\end{corollary}
\noindent Of the voting methods discussed in this paper, only Plurality, Plurality with Runoff, and Minimax violate the Condorcet loser criterion.

Finally, a voting method satisfies Independence of Smith-Dominated Alternatives (ISDA) if removing a candidate outside the Smith set does not change who wins. Stable Voting satisfies ISDA for any uniquely-weighted profile, but in non-uniquely-weighted profiles, Stable Voting may use candidates outside the Smith set to break ties.  For example, suppose that $A$, $B$, and $C$ form a perfectly symmetrical cycle: $A$ beats $B$ by 1, $B$ beats $C$ by 1, and $C$ beats $A$ by $1$. Further suppose that $A$ beats $D$ by 3, whereas $B$ and $C$ only beat $D$ by 1. Stable Voting will elect $A$ in this case, whereas if we were to first restrict to the Smith set by removing $D$, then there would be a tie between $A$, $B$, and $C$. 

\begin{restatable}{proposition}{ISDA}\label{Smith1} If $\mathbf{P}$ is uniquely weighted, and $B$ is not in the Smith set of $\mathbf{P}$, then $A$ is an SV winner in $\mathbf{P}$ if and only if $A$ is an SV winner in $\mathbf{P}_{-B}$.
\end{restatable}
\noindent For a proof of Proposition \ref{Smith1}, see Appendix \ref{Appendix}.

\begin{corollary}\label{SVSmith} For any uniquely-weighted profile $\mathbf{P}$, $SV(\mathbf{P})=SVS(\mathbf{P})$, where $SVS$ is the voting method that first eliminates all candidates outside the Smith set and then runs SV on the restricted profile.
\end{corollary}

\noindent In light of Corollary \ref{SVSmith}, if $\mathbf{P}$ is uniquely weighted, we can more efficiently compute SV winners by only listing pairs $A$~vs.~$B$ where both $A$ and $B$ belong to the Smith set. Even if $\mathbf{P}$ is not uniquely weighted, one could use SVS at the expense of sacrificing some of the tiebreaking power of~SV.

\section{Costs of Stable Voting}\label{CostsSection}

In this section, we briefly discuss some of the costs of Stable Voting. 

One is the cost of computing SV winners using the definition in Section~\ref{StableVotingSection}, which is an example of a \textit{recursive definition}: computing the SV winners in $\mathbf{P}$ involves computing the SV winners in simpler profiles of the form $\mathbf{P}_{-B}$. IRV also has a recursive definition: $A$ is an IRV winner in $\mathbf{P}$ if $A$ is the only candidate in $\mathbf{P}$ or there is some candidate $B$ with the minimal number of first-place votes in $\mathbf{P}$ such that $A$ is an IRV winner in the profile $\mathbf{P}_{-B}$.\footnote{\label{PUT}This is the ``parallel universe'' version of IRV (see~\citealt{Wangetal2019}). Many cities that use IRV prescribe that if there are multiple candidates with the minimal number of first-place votes, the candidate to be eliminated is determined by lottery (we thank Deb Otis for sharing FairVote's database of IRV rules). Then what we call an ``IRV winner'' in the text is a candidate with nonzero probability of being elected by the lottery-based version of IRV.
} However, the recursion for IRV typically terminates faster (since the recursion tree is typically not as wide as for SV). Our current computer implementation of SV can comfortably handle profiles with up to 20 candidates or larger profiles that are uniquely weighted with up to 20 candidates in the Smith set. Many elections of officers, votes on job shortlists, etc., involve no more than 20 candidates; and if there are many voters, we expect a uniquely-weighted profile in which the Smith set typically contains only a small number of ``front runners,'' even if there are over 20 candidates on the ballot. Thus, SV is practical in these contexts, though it is not currently practical in all voting contexts.

Another cost of SV is some violations---in an extremely small fraction of profiles---of voting criteria satisfied by some other voting methods.  The most important to discuss is \textit{monotonicity}. This criterion states that if $A$ wins in a profile $\mathbf{P}$, and $\mathbf{P}'$ is obtained from $\mathbf{P}$ by moving $A$ up in some voter's ranking, then $A$ should still win in $\mathbf{P}'$. Like IRV, SV can violate monotonicity. For SV, the basic reason is that moving $A$ up in a ranking  also means moving some other candidate $B$ \textit{down} in that ranking, and moving $B$ below $A$ may benefit another candidate $C$ whose closest competitor in some subelection is $B$, whereas moving $B$ below $A$ might not meaningfully benefit $A$ at all. 

Figure \ref{MonFig} shows the estimated percentages of linear profiles for 6 candidates and up to 51 voters in which IRV, a Condorcet-consistent variant of IRV known as Smith IRV (first restrict the profile to the Smith set and then apply IRV to the restricted profile), and SV violate monotonicity, meaning that there is \textit{some candidate} $A$ and \textit{some voter} $i$ in the profile such that moving $A$ up in $i$'s ranking causes $A$ to go from winner to loser or moving $A$ down in $i$'s ranking causes $A$ to go from loser to winner. Although the frequency with which SV violates monotonicity is not zero, it is minuscule compared to the frequencies for IRV and Smith~IRV. 

\begin{figure}[h]
\begin{center}
\includegraphics[scale=0.50]{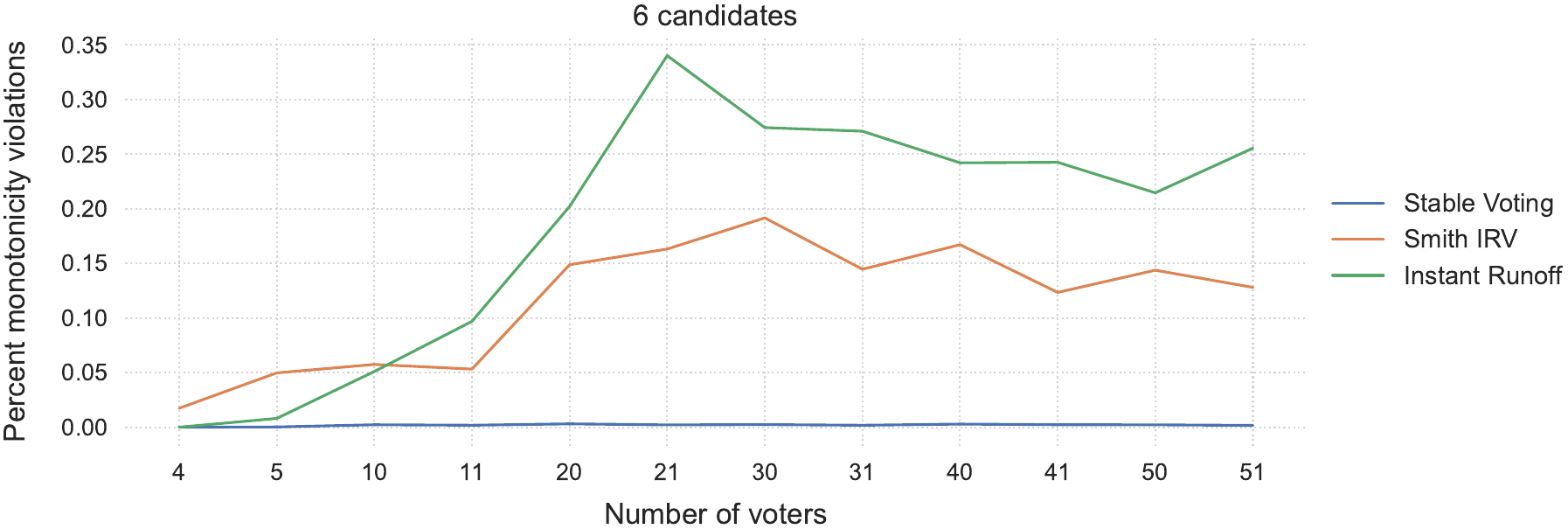}
\end{center}
\caption{estimated percentages of linear profiles witnessing monotonicity violations for IRV, Smith IRV, and SV. For each data point, we randomly sampled 50,000 profiles. More than 51 voters required too much computation time. The percentage of profiles with a single-voter monotonicity violation goes to zero for all methods as the number of voters increases.}\label{MonFig}
\end{figure}

\section{Conclusion}\label{ConclusionSection}

We have introduced Stable Voting and discussed some of its benefits and costs. Many open questions remain.  For example, are there ways of calculating SV winners that are more efficient in practice? What is the computational complexity of this problem? How does SV perform in simulations using other probability models on profiles (spatial models, urn models, Mallow's models, etc.)? What are voters' incentives for strategic voting? How do voters react to the use of SV in elections (e.g., at \href{https://stablevoting.org}{stablevoting.org})? These are only a few of the questions that must be addressed for a full assessment of Stable Voting.

\subsection*{Acknowledgements}
For helpful discussions of voting, we thank Felix Brandt, Yifeng Ding, Jobst Heitzig, Mikayla Kelley, Xavier Mora, Klaus Nehring, Chase Norman, Markus Schulze, Warren D. Smith, Andrew Souther, Nicolaus Tideman, Saam Zahedian, and Bill Zwicker. We are grateful to Nicolaus Tideman for the stimulus to write this paper.

\appendix

\section{Appendix}\label{Appendix}

The proof of Proposition \ref{Smith1} uses the following lemma. Recall the definition of the Smith set in Section~\ref{BenefitsSection} and the definition of \textit{defeat} in Section \ref{StableVotingSection}.

\begin{lemma}\label{SmithDefeat} For any profile $\mathbf{P}$ and distinct candidates $W$, $V$, and $U$ in $\mathbf{P}$, if $W$ is not in the Smith set of $\mathbf{P}$ but $V$ is, then $U$ defeats $V$ in $\mathbf{P}$ if and only if $U$ defeats $V$ in $\mathbf{P}_{-W}$.
\end{lemma}

\begin{proof} Assume that $W\not\in Smith(\mathbf{P})$ but $V\in Smith(\mathbf{P})$. We prove that $U$ does \textit{not} defeat $V$ in $\mathbf{P}$ if and only if $U$ does \textit{not} defeat $V$ in $\mathbf{P}_{-W}$. If the margin of $U$ over $V$ is not positive, then $U$ does not defeat $V$ in either $\mathbf{P}$ or $\mathbf{P}_{-W}$ (since $V,U$ is a list of candidates such that the margin of each candidate over the next is at least the margin of $U$ over $V$). So suppose the margin of $U$ over $V$ is positive. Then since $V\in Smith(\mathbf{P})$, it follows that $U\in Smith(\mathbf{P})$. Now by the definition of defeat, $U$ does not defeat $V$ in $\mathbf{P}$ (resp.~$\mathbf{P}_{-W}$) if and only if there is a list $X_1,\dots,X_n$ of candidates in $\mathbf{P}$ (resp.~$\mathbf{P}_{-W}$) with $V=X_1$ and $U=X_n$ such that for each $i<n$, the margin of $X_i$ over $X_{i+1}$ is at least the margin of $U$ over $V$; hence the margin is positive. Since a candidate outside $Smith(\mathbf{P})$ cannot have a positive margin over one in $Smith(\mathbf{P})$, and $U\in Smith(\mathbf{P})$, it follows that each $X_i$ belongs to $Smith(\mathbf{P})$. Then the existence of such a list $X_1,\dots,X_n$ of candidates in $\mathbf{P}$ is equivalent to the existence of such a list of candidates in $\mathbf{P}_{-W}$, given that $W\not\in Smith(\mathbf{P})$.\end{proof}

\ISDA*

\begin{proof} By induction on the number of candidates. First suppose $A$ is an SV winner in $\mathbf{P}$, as witnessed by $A$ vs.~$C$, i.e., this  is an earliest match in the list of matches by descending margin such that the first candidate  wins after removing the second, so $A$ is an SV winner in $\mathbf{P}_{-C}$, and the second does not defeat the first, so $C$ does not defeat $A$ in $\mathbf{P}$.  If $B=C$, then $A$ is an SV winner in $\mathbf{P}_{-B}$, so suppose $B\neq C$.  If there is a Condorcet winner (CW) in $\mathbf{P}$, then by Corollary~\ref{CondorcetLem}, $A$ is the CW in $\mathbf{P}$ and hence in $\mathbf{P}_{-B}$, so $A$ is an SV winner in $\mathbf{P}_{-B}$. Suppose there is no CW in $\mathbf{P}$. By Proposition \ref{SmithConsistency}, $A\in Smith(\mathbf{P})$. We claim that $A$ vs.~$C$ witnesses $A$ being an SV winner in $\mathbf{P}_{-B}$, i.e., $A$ vs.~$C$ is an earliest match in the list for $\mathbf{P}_{-B}$ such that the first candidate wins after removing the second, so $A$ wins in $(\mathbf{P}_{-B})_{-C}$, and the second does not defeat the first, so $C$ does not defeat $A$ in $\mathbf{P}_{-B}$. Indeed, since $C$ does not defeat $A$ in $\mathbf{P}$, $C$ does not defeat $A$ in $\mathbf{P}_{-B}$ by Lemma  \ref{SmithDefeat}. To see that $A$ wins in $(\mathbf{P}_{-B})_{-C}$, since $(\mathbf{P}_{-B})_{-C}=(\mathbf{P}_{-C})_{-B}$, we show that $A$ wins in $(\mathbf{P}_{-C})_{-B}$. Since $B\not\in Smith(\mathbf{P})$ and $C$ is not a CW in $\mathbf{P}$,  $B\not\in Smith(\mathbf{P}_{-C})$ (cf.~the proof of Proposition \ref{SmithConsistency}), so the fact that $A$ wins in $\mathbf{P}_{-C}$ and the (left-to-right direction of the) inductive hypothesis together imply that $A$ wins in $(\mathbf{P}_{-C})_{-B}$. To see that $A$ vs.~$C$ satisfies the ``earliest'' claim in  $\mathbf{P}_{-B}$, suppose for contradiction that there is a match $X$ vs.~$Y$ in  $\mathbf{P}_{-B}$ with a larger margin than $A$~vs.~$C$ such that $X$ wins in $(\mathbf{P}_{-B})_{-Y}$ and hence in $(\mathbf{P}_{-Y})_{-B}$, and $Y$ does not defeat $X$ in $\mathbf{P}_{-B}$. Since $B\not\in Smith(\mathbf{P})$, and  $Y$ is not a CW in $\mathbf{P}$, $B\not\in Smith(\mathbf{P}_{-Y})$. It follows by the (right-to-left direction of the) inductive hypothesis that $X$ wins in $\mathbf{P}_{-Y}$. Hence $X\in Smith(\mathbf{P}_{-Y})$ by Proposition \ref{SmithConsistency}, which with the fact that $Y$ is not a CW in $\mathbf{P}$ implies that $X\in Smith(\mathbf{P})$ (cf.~the proof of Proposition~\ref{SmithConsistency}). Since $B\not\in Smith(\mathbf{P})$ but $X\in Smith(\mathbf{P})$, the fact that $Y$ does not defeat $X$ in $\mathbf{P}_{-B}$ implies that $Y$ does not defeat $X$ in $\mathbf{P}$ by Lemma \ref{SmithDefeat}. But that $X$ wins in $\mathbf{P}_{-Y}$,  $Y$ does not defeat $X$ in $\mathbf{P}$, and the margin of $X$ vs. $Y$ is greater than that of $A$ vs. $C$ contradicts our initial assumption about $A$~vs.~$C$. 

Now suppose $A$ is an SV winner in $\mathbf{P}_{-B}$. Since $\mathbf{P}$ is uniquely weighted, so is $\mathbf{P}_{-B}$, so by Proposition~\ref{QRLem}, there is a unique SV winner in $\mathbf{P}_{-B}$. Hence it is $A$. Now let $A'$ be any SV winner in $\mathbf{P}$. Then as in the previous paragraph, $A'$ is an SV winner in $\mathbf{P}_{-B}$,  so $A'=A$.\end{proof}

\bibliographystyle{plainnat}
\bibliography{SV}

\end{document}